\documentclass[a4,11pt]{article}
\usepackage[margin=1in]{geometry}
\usepackage{amsmath,amsthm}
\usepackage{amssymb}
\usepackage{authblk}
\usepackage{paralist}
\usepackage[ruled]{algorithm}
\usepackage[noend]{algpseudocode}
\usepackage{hyperref}

\usepackage{epsfig}
\usepackage{enumitem}

\DeclareMathOperator*{\argmax}{arg\,max}

\newcommand{\comment}[1]{}
\newcommand{\rcast}{\text{rcast}}

\usepackage{etoolbox}\AtBeginEnvironment{algorithmic}{\small}

\comment{
	\newtheoremstyle{redstyle}% name
	{3pt}%      Space above
	{3pt}%      Space below
	{\color{black}}%         Body font
	{}%         Indent amount (empty = no indent, \parindent = para indent)
	{\color{red}\bfseries}% Thm head font
	{:}%        Punctuation after thm head
	{.5em}%     Space after thm head: " " = normal interword space;
	{}%         Thm head spec (can be left empty, meaning `normal')
	\theoremstyle{redstyle} 
}

\usepackage{color}

\newcommand{\tj}[1]{#1}

\newtheorem{definition}{Definition}
\newtheorem{theorem}{Theorem}
\newtheorem{corollary}{Corollary}
\newtheorem{lemma}{Lemma}

\newtheorem{fact}{Fact}

\newtheorem{proposition}{Proposition}

\newcommand{\labell}[1]{\label{#1}}
\newcommand{\NAT}{{\mathbb N}}

\newcommand{\cN}{{\mathcal N}}

\renewcommand{\thefootnote}{\fnsymbol{footnote}}

\newcommand{\remove}[1]{}

\newif\iffull
\fullfalse

\begin{document}

	\title{MSF and Connectivity in Limited Variants of the Congested Clique
	%Congested Clique: in the world between broadcast and unicast
		\footnote{%
			This work was supported by the Polish National Science Centre grant
			%DEC-2012/06/M/ST6/00459, 
			DEC-2012/07/B/ST6/01534.
		}
	}

	\author{Tomasz Jurdzi\'{n}ski}%\footnotemark[1]% \footnotemark[3]
	\author{Krzysztof Nowicki} %\footnotemark[1]}

	\affil{Institute of Computer Science, \\University of Wroc{\l}aw, Poland.}% \\Emails: \{tju,kno\}@cs.uni.wroc.pl}

	\date{}
	
	\maketitle              % typeset the title of the contribution
	
	\begin{abstract}
	The congested clique is a synchronous, message-passing
		model of distributed computing in which each 
		\tj{computational unit (node)} in each
		round can send message of $O(\log n)$ bits to each other node
		of the network, where $n$ is the number of nodes. This model
		has been considered under two extreme scanarios: \emph{unicast} or \emph{broadcast}.
		In the unicast model, a node can send (possibly) different message to each other node of the network. In contrast, in the broadcast model each node sends a single (the same) message to all other nodes. Following \cite{BeckerARR15}, we study the congested clique model 
		parametrized by
		the range $r$, the maximum number of different messages a node can
		send in one round. %The model with such a restriction is called $\rcast(r)$ congested clique in this paper.

	\tj{Following recent progress in design of algorihms for 
	graph connectivity and minimum spanning forest (MSF) in the
	unicast congested clique, we study these problems in limited
	variants of the congested clique.}	
		We present the first sub-logarithmic algorithm for connected components in the broadcast congested clique.
		Then, we show that efficient unicast deterministic algorithm for MSF \cite{Lotker:2003:MCO:777412.777428} and
		randomized algorithm for connected components \cite{GhaffariParter2016} can be efficiently \tj{implemented
		in the rcast model with range $r=2$, the weakest model of the congested clique above the broadcast variant ($r=1$) in the hierarchy with respect to range.} More importantly, our \tj{algorithms give the first solutions with optimal capacity of communication
		edges}, while preserving small round complexity. 
		
		\iffull
		An interesting direction arising from these results 
		is to determine a relationship between adaptiveness (the number of rounds) and communication complexity (the sum
		of sizes of transmitted messages).
		\fi
		
\comment{		
		Our main focus is on the impact of the parameter $r$ on round complexity. Firstly, we show a hierarchy result (for determinstic as well as randomized protocols) saying that the increase of $r$ increases the family of problems which can be solved in a given number
		of rounds. Secondly, we observe an exponential gap between deterministic and randomized complexity of some graph problems for a fixed
		constant value of the parameter $r$. Thirdly, we study round complexity 
			of specific problems in the rcast model. Specifically, we show that connectivity and a minimum spanning forest problem can be solved much faster than the obvious $U\cdot \log_r n$, where $U$ is the round complexity of a problem in the unicast model of the congested clique.
}
			
	\end{abstract}

	\comment{
	\vfill
	
	\small
		\noindent
		Corresponding author: {, e-mail:  \url{}
		\\
		
		\noindent

	\vfill
	
	}
	} %%% comment

	\renewcommand{\thefootnote}{\arabic{footnote}}

%		\thispagestyle{empty}
%	\setcounter{page}{0}

%	\newpage

%\newcommand{\rcast}{\text{rcast}}
\newcommand{\rrcast}{\text{rrcast}}
\newcommand{\RCAST}{\text{RCAST}}
\newcommand{\lrcast}{\text{lrcast}}
\newcommand{\lRCAST}{\text{lRCAST}}

\newcommand{\gora}[1]{\vspace*{#1pt}}

	\section{Introduction}
	\label{s:intro}
	Recently, the congested clique model of distributed computation attracted much attention in algorithmic community. In this model, each pair of $n$ nodes of a network is connected by a separate communication link. That is, the network forms an $n$-node clique.
Communication is synchronous, each node in each
		round can send message of $O(\log n)$ bits to each other node
		of the network. The main %theoretical 
		purpose of such a model
		is to understand the role of congestion in distributed
		computation.

		The congested clique model has been mainly considered in two variants: \emph{unicast} or \emph{broadcast}.
		In the unicast model, a node can send (possibly) different message to each other node of the network. In contrast, in the broadcast model each node can only send a single (the same) message to all other nodes \tj{in a round}.
		
		Following \cite{BeckerARR15}, we study the congested clique model 
		parametrized by \tj{the range} $r$,
 		the maximum number of different messages a node can
		send in one round.  
		%The parameter $r$ is called the \emph{range}. 
		We call the model with such a restriction the \emph{rcast} congested clique. (Note that $r=1$ corresponds to the broadcast congested clique and $r=n$ to the unicast congested clique.)
		
		As the broadcast and unicast models differ significantly in the amount of information which can be exchanged in a round, it is natural to introduce an intermediate model which uses a quantitative measure of usage
		of the possibility of sending different messages in each outgoing link.
The study of the rcast model is aimed at exploring this research direction.

\subsection{The model}
	We considered the congested clique model with the following parameters:
\iffull
	\begin{itemize}[noitemsep,topsep=0pt]
	\item
	\gora{5}
	$r$ -- the maximum number of different messages a node can send
	over its outgoing links in a round;
	\item
	\gora{5}
	$b$ - the maximum size in bits of a message (bandwidth); 
	%if not stated otherwise, we assume that $b=\log n$;

	\item
	\gora{5}
  $n$ -- the size of the network (number of nodes);
		
	\end{itemize}
\else
	$r$ -- the maximum number of different messages a node can send
	over its outgoing links in a round;	$b$ -- the maximum size of a message (\emph{bandwidth}); 
	%if not stated otherwise, we assume that $b=\log n$;  
	$n$ -- the number of nodes in the network/graph.
\fi	
	The model with the above parameters will be denoted $\rcast(n,r,b)$.
	Usually, we consider the model with $b=\log n$ and therefore the model
	$\rcast(n,r,\log n)$ is also denoted $\rcast(n,r)$.
	
\comment{
	%\item
	We also introduce a new parameter $p$ equal to the maximum number of different transmission patterns a node	can use over an execution of an algorithm. The model
	with such a restriction is denoted $\rcast_p(n,r,b)$. 
	
	Our interest is mainly in the scenario with $p=O(1)$ (i.e., $p$ is independent of other parameters, including the size of network) called here \emph{limited rcast}
	and denoted $\lrcast(n,r,b)=\bigcup_{i\in\NAT} \rcast_i(n,r,b)$.
}

%Given the class $\rcast_i(n,r,b)$, $\RCAST_i(n,r,b)$ denotes the class of problems
\newcommand{\round}{\text{ROUND}}
\newcommand{\rround}{\text{RROUND}}
%Let $P$ be an algorithmic problem defined on graphs. 
\comment{
The round complexity
$\round_{n,r,b}(P)$ of a problem $P$ is the number of rounds sufficient and necessary 
to solve $P$ on graphs with $n$ nodes in the model $\rcast(n,r,b)$.
}

\iffull
\subsection{Randomized algorithms}
\else
%\paragraph{Randomized algorithms}
\fi

We consider randomized algorithms in which a computational unit in each
node of the input network can use private random bits in its computation.
%The notions $\rrcast(n,r,b)$, $\rrcast(n,r)$ denote randomized counterparts
%of $\rcast(n,r,b)$ and $\rcast(n,r)$.
%
We say that some event holds with high
probability (whp) for an algorithm $A$ running on an input of size $n$  if this event holds with probability $1-1/n^c$ for a given constant $c$.

\iffull
\subsubsection{Graph problems in the congested clique model}
\else
\paragraph{Graph problems in the congested clique model}
\fi
Graph problems in the congested clique model are considered in the following framework. The joint input to the $n$ nodes of the network is an
undirected $n$-node weighted graph $G(V,E,w)$,
where each node corresponds to a node of the communication network
\tj{and weights of edges are integers of polynomial size (i.e., each weight is a bit sequence
of lenght $O(\log n)$).}
Each node $u$ initially knows the network size $n$, its unique ID in $[n]$, 
the list of IDs of its neighbors in the input graph and the weights of its incident edges.
All graph problems are considered in this paper in accordance with this definition.

\iffull
\subsubsection{Connected Components and Minimum Spanning Forest in the congested clique}
\else
%\paragraph{Connected Components and Minimum Spanning Forest}
\fi
 In the paper, we consider connected components problem (CC) and minimum spanning forest problem (MSF). 
%We assume that an input graph is a subgraph of the communication network -- each node of the input graph is mapped on a unique machine and the edges between nodes are mapped to the links between corresponding machines. 
%Each machine / node knows all weights of incident edges. 
Our goal is to compute CC or MSF of input graph, i.e.,  each node should know the set of edges inducing CC/MSF at the end of an execution of an algorithm.

\iffull
\subsection{Complexity measures}
\else
\paragraph{Complexity measures}
\fi
\tj{The key complexity measure considered in context of the congested clique models is \emph{round complexity} (called also \emph{time}) which
%$\round_{n,r,b}(P)$ of a problem $P$ is 
is equal to the number of rounds in which an algorithm works for instances of problems
of a given size.}
%sufficient and necessary 
%to solve $P$ on graphs with $n$ nodes in the model $\rcast(n,r,b)$.

\tj{For the rcast model, the \emph{range} $r$ is also a parameter determining complexity of
an algorithm. %Beside the extreme cases of 
Below, we give an observation justifying the statement that the key increase in 
communication power is between $r=1$ (broadcast congested clique) and $r=2$ (the weakest
variant of the congested clique above the broadcast model wrt the range).
}
%\iffull
%    \subsection{Unicast vs rcast}
%\else
%    \paragraph{Unicast vs rcast}
%\fi		
    \par 
    \begin{fact}
    \label{simulationTheorem}
    One round of an algorithm $A$ from $\rcast(n,n)$ may be simulated in $\rcast(n,r)$ in $O(\log_r n)$ rounds.
    \end{fact}
\iffull		
    \begin{proof}
      The local computation part of the algorithm $A$ stays the same. As for communication part, we can split each message of original protocol into blocks of size $\lfloor \log_r \rfloor$ and sent them in separate rounds.  Therefore, if in one round protocol sent message of size $\log n$, after $\frac{\log n}{\lfloor \log_r \rfloor} = O(\log_r n)$ rounds whole message would be sent to receiver, with no more than $2^{\lfloor \log_r \rfloor} \leq r$ messages per round.
    \end{proof}
\else
\fi		
		
    \begin{corollary}
    Given an algorithm $A$ solving a problem $P$ in $\rcast(n,n)$ in $R(A)$ rounds, one can build an algorithm $A'$ solving $P$ in $\rcast(n,r)$ in $O(\frac{R(A)}{\log r})$.
		%, thus we are interested in finding algorithms with better complexity.
    \end{corollary}

\tj{The above observations show that each algorithm designed for the unicast congested clique model might be simulated
in the range cast model with $O(\frac{\log n}{\log r})$ overhead. Thus, for problems of large complexity in the unicast
congested clique, the models unicast and $\rcast{2}$ seem to be very close to each other. However, for problems with
sublogarithmic round complexity in the unicast congested clique, the question about efficient rcast algorithms remains
interesting and relevant.

}

In order to provide accurate measure of the amount of information transmitted
over communication links of a network, we consider the \emph{edge capacity} measure.
The \emph{edge capacity} $\beta_A(i,n)$ is the (maximal) length (in bits) of messages which can be transmitted in the $i$th round of executions of the algorithm $A$ on graphs of size $n$. 
The (total) edge capacity $B(A,n)$ is the sum of edge capacities of all rounds, $B_A(n) = \sum\limits_i \beta_A(i)$.
As $n$ is usually known from the context, we use shorthands
$\beta_A(i)$ and $B_A$ for $\beta_A(i,n)$ and $B_A(n)$, respectively. 

%It is defines as follows. Let $\beta_A(i)$ be the maximal length of a message in the $i$th round %of executing algorithm $A$. Let $B(A) = \sum\limits_i \beta_A(i)$ be total capacity of communication edges.
For further references we make the following observation concerning edge capacity of 
algorithms solving CC and MSF. 
\begin{fact}\labell{f:capacity:lower}
Total edge capacity of each algorithm solving the connected components problem or the minimum spanning forest problem is $\Omega(\log n)$.
\end{fact}
\begin{proof}
At the end of an execution of an algorithm solving CC each node knows a partition
of the set $V$ of nodes into connected components. Information available to a node at the beginning of an algorithm has $O(n)$ bits (a characteristic
vector of the set of its neighbors). As there are $2^{\Omega(n\log n)}$
partitions of a set of size $n$, descriptions of some partitions require $\Omega(n\log n)$
bits. On the other hand, the number of bits received by a node during an execution of $A$ is 
$O(nB(A))$. Thus, in order to collect information of size $\Omega(n\log n)$ in each node,
the capacity $B(A)$ has to satisfy $B(A)=\Omega(\log n)$.

\comment{tu dowod, ktory znalazlem po napisaniu powyzszego:
    \par As there are $(n+1)^{n-1}$ possible rooted spanning forests, and each node have to know MSF as a result of an algorithm, each node must receive at least $\Theta(n \log n)$ bits during execution of algorithm. Thus we know, that any algorithm $A$ finding MSF must have $B(A) \in \Omega(\log n)$, which makes algorithm $\mathbb{L}$ non optimal with respect to $B$. We will show, that with small improvement in presented algorithm it is possible to achieve total capacity of communication edges $O(\log n)$. This version of algorithm will consist of two stages.
}
\end{proof}

\comment{
For a problem $P$, the round complexity
$\rround_{n,r,b}(P)$ is the number of rounds sufficient and necessary 
to solve $P$ on graphs with $n$ nodes in the model $\rrcast(n,r,b)$ with
high probability.
}

\comment{
    \subsection{Related results}
    Considering deterministic algorithms in $\rcast(n,n)$ currently best known solution is one proposed by Lotker et al.\ \cite{Lotker:2003:MCO:777412.777428} working in $O(\log \log n)$ round number. However, if we allow randomization best known solution for $\rcast(n,n)$ requires $O(\log^*n)$ rounds  \cite{GhaffariParter2016}. On the other hand, if we allow bandwidth of size  $\sqrt{n}\log n$, we can compute MSF in constant round number even in broadcast congested clique \cite{DBLP:journals/corr/MontealegreT16}.
}

\subsection{Related work}
The rcast model of the congested clique was introduced in \cite{BeckerARR15,BeckerARR16}. The authors presented examples showing the substantial difference between the case $r=1$ (broadcast model) and $r=2$. Moreover, it was shown that an exponential increase of the range $r$ causes $\omega(1)$ drop in round complexity for some problems. The impact of a single message size $b$ transmitted in a round through a communication link is also studied in \cite{BeckerARR15,BeckerARR16}.\iffull\footnote{In the paper, we usually assume that $b=\log n$. We state explicitly the value of $b$ each time it is not assumed $b=\log n$.}\fi

The broadcast and unicast models of congested clique were studied in several papers, e.g., \cite{Lotker:2003:MCO:777412.777428,HegemanPPSS15,GhaffariParter2016,DruckerKO13,BeckerMRT14,DBLP:journals/corr/abs-1207-1852,DBLP:journals/corr/MontealegreT16}. The recent Lenzen's \cite{DBLP:journals/corr/abs-1207-1852} constant time
routing and sorting algorithm in the unicast congested clique shows the power of the unicast model. (The routing problem according to the definition from  \cite{DBLP:journals/corr/abs-1207-1852} trivially requires $\Omega(n)$ rounds in the broadcast congested clique.) 
%Lotker et al.\ \cite{Lotker:2003:MCO:777412.777428} designed 
Lotker et al.\ \cite{Lotker:2003:MCO:777412.777428} designed a $O(\log \log n)$ round deterministic algorithm for MSF (minimum spanning forest) in the unicast model. 
\tj{Recently, an alternative algorithm of the same complexity has been presented \cite{Korhonen16}.}
The best known randomized solution for MSF in the unicast model works in $O(\log^*n)$ rounds  \cite{GhaffariParter2016}, improving the recent $O(\log\log\log n)$ bound \cite{HegemanPPSS15}. \tj{Reduction of the number of transmitted messages in the MST
algorithms was studied in \cite{DBLP:conf/fsttcs/PS16}.}
\iffull
The tradeoff between the number of rounds and the size of messages has also been
studied recently for various problems.
\fi
If messages can have $\sqrt{n}\log n$ bits (bandwidths $b=\sqrt{n}\log n$), one can compute MSF in constant number of rounds, even in the broadcast congested clique \cite{DBLP:journals/corr/MontealegreT16}.
In \cite{DruckerKO13} the authors proved that it is possible to simulate powerful
classes of bounded-depth circuits in the unicast congested clique, which points out to the power of this model and explains difficulty in obtaining lower bounds for this model. In \cite{BeckerMRT14}, randomized variants of the broadcast congested clique are considered.

Apart from purely theoretical and algorithmic interest in the congested clique, the model also 
relates to other models of processing of large-scale graphs, e.g., the $k$-machine model
\cite{KlauckNP015}, MapReduce \cite{HegemanP15} and the concept of overlay networks.

\subsection{Our results}
\comment{We show hierarchy result (for deterministic as well as randomized protocols) saying that the increase of the range parameter $r$ increases the family of problems which can be solved in a given number
of rounds. More precisely, there exists constants $c,c'>1$ and a problem $P$ such that the round complexity of $P$ for the range $r$ is at least $c'$ times smaller that the round complexity for $r'$, if $r\ge cr'$.
		
Moreover, we observe an exponential gap between deterministic and randomized complexity of specific problems for a fixed
		constant value of the parameter $r$. 
		}
		
\comment{		
Finally, we study round complexity 
			of connected components and minimum spanning forest (MSF)problem in the rcast model. Specifically, we show that the Lotker et al.\ algorithm for the unicast congested clique can be applied for computing connected components in the rcast model with range $r=2$ without loss in round complexity. We provide also an algorithm for MSF which breaks the $\log n$ bound from the broadcast model for $r=\omega(1)$.
}

		We present the first sub-logarithmc algorithm for connected components in the broadcast congested clique.
		Our algorithm works in $O(\log n/\log\log n)$ rounds and scales to models with varying sizes of messages.
		Then, we show that efficient unicast deterministic algorithm for MST \cite{Lotker:2003:MCO:777412.777428} and
		randomized algorithm for connected components \cite{GhaffariParter2016} can be efficiently adjusted to the 
		$\rcast(2)$ model, the weakest variant above the broadcast congested clique in the hierarchy of $\rcast(r)$ models
		for $r>1$. More importantly, our result imply solutions with efficient (optimal, in some case) capacity of communication
		edges, while preserving small round complexity. An interesting direction arising from these results 
		is to determine a relationship between adaptiveness (the number of rounds) and communication complexity (sum
		of sizes of transmitted messages).

%			connectivity and a minimum spanning forest problem can be solved much faster than the obvious $U\times \log_r n$, where $U$ is the round complexity of a problem in the unicast model of the congested clique.

% 	\section{Hierarchy results}

	\section{Graph terminology and tools for capacity/range reduction}\labell{sub:s:tools}

\iffull
In this section we provide some terminology and tools useful
for reduction of capacity and range of 
congested clique algorithms.
\fi

\tj{Given a natural number $p$, $[p]$ denotes the set
$\{1,2,\ldots,p\}$.}

\iffull
\subsection{Graph terminology for MST/CC algorithms}
\fi

For a graph $G(V,E)$ and $E'\subseteq E$, $C_1, C_2, ..., C_k \subset V$ is a partition of $G$ into \emph{components} with
respect to $E'\subseteq E$ if $C_i$s are pairwise disjoint, $\bigcup_{i\in[k]} C_i = V$, each $C_i$ is connected with respect to the edges from $E'$ and there are no
edges $(u,v)\in E'$ such that $u\in C_i$ and $v\in C_j$ for $i\neq j$. That is,
$C_1,\ldots, C_k$ are connected components of $G(V,E')$.

\emph{Fragment} of a graph $G(V,E,w)$ %with respect to $E'\subseteq E$ is
is a tree $F$ which is a subgraph of a minimum spanning forest of $G$.
A family $\mathbb{F}$ of fragments of $G(V,E,w)$ 
is a \emph{partition} of $G$ into fragments with respect to $E'\subseteq E$
if $F_1$ and $F_2$ have disjoint sets of nodes for each $F_1\neq F_2$ from $\mathbb{F}$,
each $v\in V$ belongs to some $F\in\mathbb{F}$ and each edge of each 
tree $F\in\mathbb{F}$ belongs to $E'$.

Given a partition $\mathcal{C}$ ($\mathcal{F}$, resp.) of a graph $G(V,E)$ into components (fragments, resp.) and $v\in V$,
$C^v$ ($F^v$) denotes the component (the fragment, resp.) containing $v$.

\comment{
    \begin{definition}
    $ $
    \begin{itemize}
      \item \textbf{Components} of a graph $G(V,E)$ with respect to $E'\subset E$ are $C_1, C_2, ..., C_k \subset V$ such that sets $C_i$ are disjoint, $\bigcup C_i = V$ and each of $C_i$ is connected with respect to edges from $E'$.
      \item \textbf{Fragments} of graph $G$ are \textbf{components} with respect to set of edges $E'$ which is a subset of the edges of the minimal spanning forest of graph $G$.
      \item By $C^v$ and $F^v$ we will denote respectively a component containing node $v$ and a fragment containing node $v$
    \end{itemize}
    \end{definition}	
}		
    %In the following sections we 
		We will usually consider components with respect to a set of edges which are known to all nodes in congested clique.  

\tj{We say that a fragment (component, resp.) is \emph{growable} if there is an edge connecting it with some other fragment/component in the considered graph.
An edge $(u,v)$ is \emph{incident} to a fragment $F$ (component $C$, resp.) wrt to some partition of a graph in fragments/components
if it connects $F$ with another fragment (component, resp.), i.e., $F^u\neq F^v=F$ or $F^v\neq F^u=F$ ($C^u\neq C^v=C$ or $C^v\neq C^u=C$, resp.).}
\iffull
\subsection{Tools for capacity/range reduction}
\else
\subsubsection{Tools for capacity and range reduction}
\fi
As tools to reduce edge capacity and range of congested clique algorithms,
we introduce the \emph{local broadcast} problem and the \emph{global broadcast} problem. 
In the local broadcast problem, the following
parameters are known to each node of a network
    \begin{itemize}[noitemsep,topsep=0pt]
    \item a set $T\subset V$,
    \item a set $R\subset V$,
	\item a natural number $b$.
%     \item for each $i$, some nodes $v \in V'_i \subseteq V_i$ has a message $m^v$ of size $p_i$, sets $V'_i$ are known to the whole network
     \end{itemize}
Moreover, each node $v\in T$ has its own message $M^u$ of length $b$.
As a result of local broadcast, each node $v \in R$ receives the message $M^u$ from each $u \in T$. 

\begin{proposition}\labell{p:local:broadcast}
Algorithm~\ref{broadcastInFragment} (LocalBroadcast) solves  the local broadcast problem in $O(1)$ rounds with range $r=2$ and capacity $1$,
provided $|T|b=O(n)$. It is possible to execute LocalBroadcast (Algorithm~\ref{broadcastInFragment}) simultaneously for $k$ triplets $(T_i,R_i, b_i)_{i \in[k]}$, as long as $T_i$'s are pairwise disjoint, $R_i$'s are pairwise disjoint and $|T_i|b_i \in O(n)$ for each $i\in[k]$.
\end{proposition} 
\begin{proof}
First, we show that it is possible to send messages from $T$ to $R$ in two rounds using one bit per communication link per round, provided $|T| b\leq n$. Let us split nodes $V$ into $t=|T|$ segments $S_1,\ldots,S_{t}$ of size $b$. %The algorithm works as follows. 
In Round 1, each $v\in T$ sends the $j$th bit ($j\in[b]$) of its message $M^v$ to the $j$
node of its segment. In Round 2, each node $v$ sends the bit received in Round~1 to all nodes from $R$.
    \begin{algorithm}[h]
      \caption{LocalBroadcast$(T, R, b)$} %\Comment{local broadcast from $T_i$ to $R_i$}}
      \label{broadcastInFragment}
      \begin{algorithmic}[1]
        \State assign a segment \tj{$S_i$} of nodes of size $b$ to each $v\in T$
        \State Round 1: each node $v$ sends the $j$th bit of $M^v$ to the $j$th node of its segment
        \State Round 2: each node $u$ (from the segment assigned to $v$) sends the bit received in Round~1 to all nodes from $R$
      \end{algorithmic}  
    \end{algorithm}
Algorithm~\ref{broadcastInFragment} solves the problem for one pair $(T, R)$ of transmitters and receivers. However, it is possible to solve it simultaneously for multiple pairs $(T_i, R_i)_{i\in [k]}$ with messages of size $b_i$, as long as $T_i$ are pairwise disjoint, $R_i$ are pairwise disjoint and $|T_i|b_i\in O(n)$.
Observe that all transmitters in Round~1 belong to $T_i$.
As $T_i$'s are pairwise disjoint, Round~1 can be done simultaneously for each $i\in[k]$.
On the other hand, all receivers in Round~2 belong to $R_i$. As $R_i$'s are pairwise disjoint, Round~2 can be done simultaneously for each $i\in[k]$.
Finally, if $|T_i|b_i\leq c\cdot n$ for a constant $c>1$, we can solve the local broadcast problem in $2c$ rounds by repeating Algorithm~\ref{broadcastInFragment} $c$ times. 
	
As transmission at each edge in Round~1 and Round~2 contains $1$ bit, we obtain a solution
with range $2$ and edge capacity $1$ in time $O(1)$.
%    \par Using this procedure several times we can say that if each node $v \in V_i$ has message of length $p_i \in O(\frac{n}{|V_i|})$, then it is possible to solve this problem in some constant number of rounds, each using one bit per communication link. Therefore this procedure is implementable in rcast(n,2).
\end{proof}

\noindent\textit{Remark.}
\tj{Design of algorithms in the unicast congested clique has been recently fostered by the Lenzen's routing lemma \cite{DBLP:journals/corr/abs-1207-1852}. 
For a reader familiar with Lenzen's paper, Proposition~\ref{p:local:broadcast} might seem to be a corollary
from his result.
We remark here that it is not the case, because the
overall size of all copies of a message $M^v$ for $v\in T_i$ is $b_i|R_i|$ which might be $\omega(n)$.
%
%If $b_i$, $T_i$ are $R_i$ are known globally,
%Proposition~\ref{p:local:broadcast} can be obtained
%through Lenzen's general technique. We have given a direct proof
%as the specific task of the local broadcast can be solved relatively
%simple, while the solution illustrates the power of the unicast
%congested clique model.
\comment{For our purposes, one can 
show the following corollary from this result. Assume that, for each $i\in[n]$, 
the node $v_i$ has messages $M_i^1, \ldots, M_i^n$ such that
$M_i^j$ should be delivered to $v_j$. Then, if the sum of lengths
of messages $|M_i^1|+ \ldots+ |M_i^n|$ is at most $n$ for each $i\in[n]$, all messages can be delivered in $O(1)$ rounds with edge
capacity $1$ and range $r=2$. 
The local broadcast problem not always satisfies the condition that the number of bits
received by a node is at most $n$. Therefore we need a separate lemma.
}
}

Next, we define the \emph{global broadcast problem.}
%whose solution is another tool for reduction of the range and edge capacity of congested clique algorithms.
Assume that each node from a set $S\subseteq V$ of nodes knows (the same) message $M$
of length $b$.
The \emph{global broadcast problem} is to deliver $M$ to each node $v\in V$
of the network.

%\tj{cos o rutingu lenzena wspomniec?}

\begin{proposition}\labell{p:broadcast}
The global broadcast problem can be solved in one round with range $r=1$ and \tj{edge} capacity $\lfloor\frac{b}{|S|}\rfloor$.
\end{proposition}
\begin{proof}
The problem can be solved by splitting the common message into $\lfloor\frac{b}{|S|}\rfloor$ parts assigning the $i$th part to the $i$th element of $X$ for $i\in[\lfloor\frac{b}{|S|}\rfloor]$.
Then, each $v\in X$ broadcasts its part to the whole network.
\end{proof}

	\section{Connected components in the broadcast congested clique}
	This section is devoted to the broadcast congested clique, the weakest variant of the 
congested clique model.
%
%First, we introduce some useful notations. Then, 
First, we recall a distributed
implementation of the well known Boruvka's algorithm for MST. 
Then, we design a new algorithm for connectivity
which (unexpectedly?) shows that the $\log n$ bound on round
complexity can be broken in the broadcast congested clique.
%, tj{i.e., in the weakest variant
%of the congested clique model}.

    \subsection{Minimum spanning forest in broadcast congested clique} % ($\rcast(n,1)$)}
    Minimum spanning forest can be computed using a distributed version of the classical Boruvka's algorithm. The algorithm works in \emph{phases}. At the beginning of phase $i$ a partition
		$\mathcal{F}$ into fragments of size $\ge 2^i$ is given.
    During phase $i$ %some edges are announced and we use it to merge separate trees. Let us define:
new fragments of size $\ge 2^{i+1}$ are determined, based on the lightest edges incident to all
fragments.

    \par In the distributed implementation of the Boruvka's algorithm each node knows the set of fragments at the beginning of a phase. During the phase each node $v$ announces (broadcasts) the lightest edge connecting $v$ with a node $u \not \in F^v$. 
		Using those edges, each node can individually (locally) perform the next phase of
		the Boruvka's algorithm and determine \tj{new (larger) fragments.}
%		We use those edges to merge separate fragments. This algorithm can be implemented in the broadcast congested clique in $O(\log n)$ rounds.
    \begin{theorem}
      \label{BroadcastCC}
      Boruvka's algorithm %in $\rcast(n,1)$ 
			can be implemented in $O(\log n)$ rounds in $\rcast(n,1)$.
    \end{theorem}
\comment{    
		\begin{proof}
    \par After $i$th phase, each fragment contains at least $2^i$ nodes, as at the beginning we have fragments of size 1 and in each phase we merge each fragment with at least one other fragment. Therefore, after $O(\log n)$ communication rounds we will have components of size $n$, which means we will have MSF of given graph. 
    Therefore, Theorem \ref{BROADCAST_ALG} is correct.
    \end{proof}
}    
    \subsection{Connected components algorithm}
     To calculate connected components we could use the standard Boruvka's algorithm as well. 
		However, we are not forced to select the lightest edge incident to each component.
		%therefore we can select them more carefully and provide some improvment in the number of required rounds. 
		Our general idea is to prefer edges which connect nodes to components of large degree.
		And the intended result of a phase should be that each component either has small degree or it \tj{is connected} to some
		``host'' of large degree (directly or by a path of length larger than one). As the number
		of such ``hosts'' will be relatively small, we obtain significant reduction of the number of
		components of large degree in each phase. Moreover, we separately deal with components of small degree
		by allowing them to broadcast all their neighbours at the final stage of the algorithm.

		More precisely, we define $\deg(v)$ for a vertex $v$ wrt a partition $\mathcal{C}$
		as the number of components connected with $v$, i.e.,	$\deg(v)=|N(v)|$, where
%		
%		Given a partition $\mathcal{C}$ into components, let $N(v)$ denote the set of neighbouring components
%		of $v$, i.e., 
		$$N(v)=\{C\in \mathcal{C}\,|\, \exists u \in C \mbox{ such that } (v,u) \in E\mbox{ and }C\neq C^v\}.$$
%
%
%		$$\deg(v)=|\{C'\in\mathcal{C}\,|\, (v,u)\in E\mbox{ for some } u \mbox{ such that }C^v\neq C^u\}.$$
%\tj{CZY POZNIEJ UZYWAMY PODOBNEJ NOTACJI DO $N(v)$???}
		For a component $C\in\mathcal{C}$, $\deg(C)=\max_{v\in C}\{\deg(v)\}$. Note that,
		according to this definition, the degree of a component $C$ might be smaller than
		the actual number of components containing nodes connected by an edge with nodes
		from $C$. Our definition of degree is adjusted to make it possible that
		degrees of components can be determined in $O(1)$ rounds.
				
		The algorithm is parametrized by a natural number $s$ which (intuitively) sets
		the threshold between components of small degree (smaller than $s$) and large degree
		(at least $s$).
		Given a partition $\mathcal{C}$ of the graph \tj{into components (wrt edges known to all nodes)}, we define the linear ordering
		$\succ$ of components, where $C\succ C'$ iff $\deg(C)>\deg(C')$ or $\deg(C)=\deg(C')$ and $\text{ID}(C)>\text{ID}(C')$.
A component $C$ is a \emph{local maximum} if all its neighbors
are smaller with respect to the $\succ$ ordering.		
		
Our algorithm consists of the main part and the \emph{playoff}.
The main part is split into \emph{phases}.
At the beginning of phase $1$ each node is \emph{active}
and it 
forms a separate %\emph{active} 
component.
\tj{During an execution of the algorithm, nodes from non growable
components and components of small degree (smaller than $s$)
are \emph{deactivated}.}
At the beginning of a phase, a partition of the graph
of active nodes is known to the whole network.
%The partition of a graph
%into components with respect to already broadcasted edges is build
%at the end of each phase. 
%Components of small degree are deactivated at the end of each phase.
\tj{First, each node $v$ determines $N(v)$ and announces its degree $\deg(v)$
wrt the current partition of the set of active nodes into components (Round 1). 
With this information, each node $v$ knows the ordering of components of the
graph of active nodes according to $\succ$.
Then, each \emph{active} node $v$ (except of members of local maxima) broadcasts
its incident edge to the largest active component from $N(v)$ according to $\succ$
relation (Round~2). }
Next, each node $v$ of each local maximum $C$ checks whether edges connecting $C$ 
to all components containing neighbors of $v$ \tj{(i.e., to components from $N(v)$)} 
have been already broadcasted. 
If it is not the case, an edge connecting $v$ to a new component $C'$ 
(i.e., to such $C'$ that no edge connecting $C$ and $C'$ was
known before) is
broadcasted by $v$ (Round~3).
Based on broadcasted edges, new components are determined
and their degrees are computed (Round~4). Each new component with degree
smaller than $s$ is \emph{deactivated} at the end of a phase.
%

%Importantly, 
%
The playoff lasts $s$ rounds in which each node $v$ of each
deactivated component broadcasts edges going to all
components connected to $v$ (there are at most $s$ such components for
each deactivated node).
More precise description of this strategy is presented
as Algorithm~\ref{BRCC}. 
\tj{The key property for an analysis of complexity of our algorithm
is that each active component $C$ of large degree is either connected
during a phase to all its neighbors or to a component which is larger
than $C$ according to $\succ$.
}

    \begin{algorithm}
      \caption{BroadcastCC$(v,s)$\Comment{$s$ is the threshold between small/large degree}}
      \label{BRCC}
      \begin{algorithmic}[1]
        \While{there are active components} \Comment{execution at a node $v$} %{$E \neq \emptyset$}
          \State \textbf{Round~1:} $v$ broadcasts $\deg(v)$
				\If{$\deg(v)>0$}
					\State $C_{\text{max}}(v)\gets$ the largest element of $N(v)$ wrt the ordering $\succ$
					\State \textbf{Round~2:} 
					\State \label{line:b1}\textbf{if} $C^v$ is not a local maximum \textbf{then} $v$ broadcast an edge $(u,v)$ such that $u\in C_{\text{max}}$
%          \State \textbf{each node $v$ announces edge to a component $c = \argmax\limits_{x \in N(v)}\big( \max\limits_{v\in x} \deg(v) \big) \Leftrightarrow \deg(v) \leq \max\limits_{u\in c} \deg(u)$}
          \State \textbf{Round~3:}
					\If{$C^v$ is a local maximum}
						\State $N_{\text{lost}}(v)\gets\{C\,|\, C\in N(v)\mbox{ and no edge connecting }C\mbox{ and } C^v \mbox{ was broadcasted}\}$
						\If{$N_{\text{lost}}(v)\neq\emptyset$}
							\State $u\gets$ a neighbor of $v$ such that $u\in C$ for some $C\in N_{\text{lost}}(v)$
							\State \label{line:b2}$v$ broadcasts an edge $(u,v)$
						\EndIf
					\EndIf
%					\State \textbf{let $N_{\mathit{old}}(v) \gets N(v)$}
%          \State \textbf{let $\deg_{\mathit{old}}(v) \gets \deg(v)$}
				\EndIf
          \State {$v$ computes the new partition into components, using broadcasted edges}
          \State \textbf{Round~4:} $v$ broadcasts $\deg(v)$\Comment{degrees wrt the new components!}
					\State \textbf{if} $\deg(C^v)< s$ \textbf{then} deactivate $v$
%           \Comment each node $v$ knows which edges have second end in $C^v$, thus can exclude it from further computation
        \EndWhile
        \State \textbf{Playoff ($s$ rounds)}: deactivated nodes broadcast edges to neighboring components.
      \end{algorithmic}
    \end{algorithm}

    \begin{theorem} 
    \label{BRCCTheorem}
    Algorithm \ref{BRCC} solves the spanning forest problem in $O(s + \log_s n)$ rounds for an $n$-node graph.
    \end{theorem}
    \begin{proof}
		First, consider round complexity of the algorithm. It is clear that Playoff has $s$ rounds.
		To show the claimed complexity we show that the number of active components is decreased
		at least $s$ times in each phase. An intuition is that all components join with (some) local maxima
		and thus each local maximum of large degree ``combines'' at least $s$ components
		\tj{in a new, larger component}.
		However, the situation is not that simple, as there might be many local maxima.
		
		In order to formalize the intuition, consider a directed graph $G_{\text{phase}}$ of components active
		at the beginning of a phase, where $(C_1, C_2)$ is an edge 
		in $G_{\text{phase}}$ \textbf{iff} 
		%\begin{itemize}
		%\item 
		a node from $C_1$ broadcasts
		an edge connecting it with $C_2$ in step~\ref{line:b1} of the phase (edge of type 1) \textbf{or} 
		%\item 
		$C_1$ is a local maximum, a node from $C_1$ broadcasts an edge connecting it with
		some $C'$ in step~\ref{line:b2}, while a node from $C'$ broadcasts an edge connecting it with
		$C_2$ in step~\ref{line:b1} (we call it edge of type 2).
		%%\end{itemize}. 
		
		\noindent The algorithm guarantees that 
		\begin{enumerate}[label=(\alph*), noitemsep,topsep=0pt]
		\item $G_{\text{phase}}$ is acyclic. 
		
		Indeed, each edge $(C_1, C_2)$ resulted from broadcasts 
		in step~\ref{line:b1} satisfies $C_1\prec C_2$. Moreover,
		an edge is broadcasted from $C_1$ to $C'$ in step~\ref{line:b2} iff
		all nodes from $C'$ broadcasted connections to components larger than
		$C_1$ wrt  $\succ$ ordering.
	
		\item
		Each connected component $C$ (i.e., each node of $G_{\text{phase}}$)
		is either a sink of $G_{\text{phase}}$ connected with (at least) $\deg(C)$ nodes in $G_{\text{phase}}$ or has out-degree at least one.
		
		This property follows from the fact that only nodes of local maxima are candidates for sinks,
		as only they do not broadcast
		in step~\ref{line:b1}. Moreover, assume that $C$ is a local maximum and there is a neighbor $C'$ of $C$
		whose nodes have not broadcasted connections with $C$ in step~\ref{line:b1}. Then a node(s)
		from $C$ broadcast in step~\ref{line:b2} which implies that out-degree of $C$ is at least one.
		
		\item
		Each connected component of a partition obtained at the end of a
		phase contains at least one sink of $G_{\text{phase}}$.
		
		If one ignores that edges of $G_{\text{phase}}$ are directed then certainly
		new components at the end of the phase correspond to connected
		components of $G_{\text{phase}}$. This follows from the fact that
		edges of $G_{\text{phase}}$ correspond to connections between 
		components (by an edge or a path of two edges \tj{in the original graph}) broadcasted during
		the phase. As $G_{\text{phase}}$ is acyclic, each connected component
		contains a sink.
		\end{enumerate}
 	
		Let $\mathcal{C}$ be a partition into components at the beginning of a phase
		and $\mathcal{C}'$ be the partition into components at the end of that phase,
		before deactivating components of small degree.\footnote{Note that deactivation of components
		of degree $<s$ at the end of a phase does not guarantee that degrees of all components
		are $\ge s$ at the beginning of the next phase. This is caused by the fact that deactivation
		of some components might decrease degrees of components which remain active
		%(as we calculate
		(degrees are calculated only among active nodes).}
		The above observations imply that each component of $\mathcal{C}'$ either
		contains only components of $\mathcal{C}$ of small degree (smaller than $s$) or it contains
		at least $s+1$ components from $\mathcal{C}$. Contrary, assume that a component $C'$
		of $\mathcal{C}'$ contains a component $C\in \mathcal{C}$ of degree $\ge s$, while
		$C'$ contains altogether at most $s$ components of $\mathcal{C}$. Then,
		there is a directed path from $C$ to a sink $C_{\text{sink}}$ of degree at least $\deg(C)\ge s$.
		Property (b) implies that at least $s$ components of $\mathcal{C}$ have edges
		towards $C_{\text{sink}}$ in $G_{\text{phase}}$. This contradicts the contrary assumption
		that $C'$ contains altogether less than $s$ components of $\mathcal{C}$.
				
		Summarizing, assume that we have $p$ active components at the beginning of a phase.
		Then, at the end of the phase, there are at most $p/s$ new components which contain
		at least one component whose degree at the beginning of the phase was $\ge s$.
		It remains to consider the final components of the phase which are composed only from components
		whose degree was $<s$ at the beginning of the stage. However, as the degree of a node
		cannot increase during the algorithm, the degrees of these new components are $<s$
		and they are deactivated at the end of the phase. Thus, each phase decreases the number
		of active components at least $s$ times -- there are at most $\log_s n$ phases.
		
		Correctness of the algorithm follows from the fact that each node of each
		deactivated component
		can broadcast its connections with all other components during Playoff.
		\tj{Moreover, active components are connected subgraphs of $G$ at each stage.}

    \end{proof}
The minimum of $s + \log_s n$ is obtained for $s=\frac{\log n}{\log \log n}$.
Then, Algorithm \ref{BRCC} works in $O(\log n/\log\log n)$ rounds. 
    \begin{corollary}
    It is possible to solve the connected components problem in the broadcast congested clique in $O\left(\frac{\log n}{\log \log n}\right)$ rounds.
    \end{corollary}
    %
%    \par This result is an improvement over the standard Boruvka's algorithm which, up to our knowledge, was so far the best known solution for the broadcast congested clique. 
\tj{Now, consider the model in which the maximum size of a message (bandwidth) is larger than $\log n$. If $s=d$ in Algorithm \ref{BRCC}, we get $\log_d n$ phases, each requiring $O(\log n)$ bits per node. 
%
	%	After that there is part, when we are announcing 
	Edges from deactivated nodes are broadcasted during Playoff in one round, using
	$O(d\log n)$ bits. %With bandwidth $d \log n$ it can be done in one round instead of $s$, 
	This gives $O(\log_d n)$ round algorithm using $O(\log n(d + \frac{\log n}{\log d}))$ bits per node during the whole execution.}
    \begin{corollary}
    It is possible to solve connectivity problem 
		in the broadcast congested clique with bandwidth $d \log n$ 
		using 
				$\log_d n$ rounds and $O(\log n(d + \frac{\log n}{\log d}))$ bits \tj{transmitted
		by each} node. 
    \end{corollary}
    \tj{
%    
    %\par 
		The above corollary gives an improvement over a result from \cite{MT2016}, where the total
		number of bits per node is $O(d\frac{\log^2 n}{\log d})$ in $O(\log_d n)$ rounds.
		%as the total number of bits sent by one node is $O(\log n(d + \frac{\log n}{\log d}))$ instead of 
%		$O(d\frac{\log^2 n}{\log d})$. 
Moreover, our algorithm is simpler than that in \cite{MT2016}, since it does not require number 
theoretic techniques as $d$-pruning and deterministic sparse linear sketches.
    }

	\section{Deterministic rcast algorithm for minimum spanning forest}
	In this section we provide a deterministic algorithm for minimum spanning forest (MSF) in the rcast model. First, we describe a generic algorithm for minimum spanning tree from \cite{Lotker:2003:MCO:777412.777428}. Then we provide a new efficient $\rcast(n,2)$ version of this
general algorithm. Finally, an algorithm optimizing the range $r$ and
achieving asymptotically optimal edge capacity is presented. 

\subsection{Generic MSF Algorithm}
First, we introduce terminology useful in describing (distributed) algorithms for MSF.

For a graph $G(V,E,w)$ and its partition into fragments, we say that an edge $e = (v,u)$ is \emph{relevant} for a set $A\subseteq V$ if $F^v\neq F^u$ and $e$ is the lightest edge connecting a node from $A$ and a node from fragment $F^u$.
Let $E_{A,\mu}$ denote the set of $\mu$ lightest \tj{relevant} edges incident to the set $A\subset V$.
Moreover, $\cN_{F,\mu}$ for a fragment $F$ denotes the set of fragments connected with $F$ by edges from
$E_{F,\mu}$.

\comment{
\begin{definition}
$ $\begin{itemize}[noitemsep,topsep=0pt]
    \item fragment / component is growable if it has edges connecting it with some other fragment / component
    \item for a given graph $G$ and its partition into fragments we will say, that for a set $V$ an edge $e = (v,u)$ is \textbf{relevant} if it is the lightest edge connecting a node from set $V$ and fragment $F_u$
    \item let us denote set of $\mu$ lightest relevant edges for a set $V$  by $E_{V,\mu}$
    \item let us denote set of components connected to a component $C$ via edges from $E_{C,\mu}$ by $\cN_{C,\mu}$ 
  \end{itemize}
\end{definition}
}

%In the first place let us formulate lemma, which is the main part of presented MSF algorithms.
Below, we give a lemma which is crucial for the first efficient unicast congested clique
algorithm for MSF \cite{Lotker:2003:MCO:777412.777428}.
\begin{lemma}
\labell{lotkerLemma}
\cite{Lotker:2003:MCO:777412.777428} 
Let $\mathcal{F}$ be a partition of a graph $G(V,E)$ into fragments,
let $E_{\mathcal{F}}$ be the set of edges in the trees of the partition $\mathcal{F}$.
Then, for each $\mu>0$, the minimum spanning forest $\mathcal{F}'$ of $G(V,E_{\mathcal{F}}\cup\bigcup_{F\in\mathcal{F}}E_{F,\mu})$
is a partition of $G(V,E)$ into fragments, such that the size of each growable tree
of $\mathcal{F}'$ is at least 
$(\mu+1) \min\limits_{F \in \mathbb{F}} |F|$.
%
%For a given set of fragments $\mathbb{F}$ let us build MSF consisting those fragments and edges $\bigcup\limits_{F \in \mathbb{F}} E_{(F,\min(\mu, \deg(F))}$. Let us select from trees forming this MSF the set of all growable fragments $\mathbb{F}'$. Then there is the following inequality $\mu \min\limits_{F \in \mathbb{F}} |F| \leq \min\limits_{F \in \mathbb{F'}} |F|$.
\end{lemma}
%This lemma has been proven in \cite{Lotker:2003:MCO:777412.777428}. 
In other words, the above lemma says that, in order to increase the size of fragments
$\mu+1$ times, it is sufficient to consider $\mu$ lightest relevant edges \tj{for}
each fragment.

Using Lemma~\ref{lotkerLemma}, one can build MSF in \emph{phases} using
the idea described in Algorithm~\ref{RCAST_MSF}.
First, let us fix
a sequence $\mu_1,\mu_2,\ldots$ of natural numbers.
Phase $i$ starts from a partition of the input graph into fragments and
ends with a new partition into larger fragments.
Before the first phase, each node is considered as a separate fragment. 
At the beginning of the $i$th phase, the set $E_{F,\mu_i}$ of $\mu_i$ lightest
relevant edges (or all relevant edges, if there are at most $\mu_i$)
is determined for each fragment $F$ of the current partition. 
Then, this information is broadcasted to all nodes of the network.
Using Lemma~\ref{lotkerLemma}, each node can compute a new partition
into fragments such that the size of the smalles growable fragment
is increased at least $\mu_i +1$ times.

Lotker et al.\ designed a congested clique implementation of Algorithm~\ref{RCAST_MSF}
which guarantees that each phase works in $O(1)$ rounds for the sequence
$\mu_1=1$ and $\mu_{i}=\mu_{i-1}(\mu_{i-1}+1)$ for $i>1$. As $\mu_k\geq n$
for $k=O(\log\log n)$, their algorithm works in $O(\log\log n)$ rounds.

%\par A generic MSF algorithm works in phases, in each phase it reduces the number of growable fragments of the MST according to the Lemma \ref{lotkerLemma}. In particular, it was shown that if in the $i$th phase we will set $\mu =f_i$, where $f_1 = 1,\ f_i = (f_{i-1}+1)f_{i-1}$, then it is possible to implement it in unicast congested clique and it will require $O(\log \log n)$ rounds. 

    \begin{algorithm}
      \caption{Minimum Spanning Forest}
      \label{RCAST_MSF}
      \begin{algorithmic}[1]
        \State{$i \gets 1$}
				\State $\mathbb{F}=\{\{v_1\}, \{v_2\},\ldots,\{v_n\}\}$
        \While{$E \neq \emptyset$}
          \State \label{gen:s1} SelectEdges($\mu_i,\mathbb{F}$)
          \State \label{gen:s2} announce edges from $E_{F,\mu_i}$
          \State \label{gen:s3} locally merge fragments, modify $\mathbb{F}$ appropriately
          \State{$E \gets E \setminus \{(u,v) | F^v = F^u \} $}
          \State{$i \gets i+1$}
        \EndWhile
      \end{algorithmic}
    \end{algorithm}
In order to illustrate problems with design of algorithms with limited
range and edge capacity, we first shortly describe the $O(\log\log n)$
solution for MSF from \cite{Lotker:2003:MCO:777412.777428}.

The selection of $\mu_i$ lightest edges incident to each fragment in
a phase (step \ref{gen:s1} of Alg.~\ref{RCAST_MSF})
is done after one round of communication as follows. For each node $v$ and each fragment $F\neq F_v$,
$v$ sends the lightest edge from the set $\{(v,u)\,|\, u\in F\}$ to all
nodes from $F$. Thus, the edge capacity $\Theta(\log n)$ is needed
in each phase. The upper bound on the range is equal to the number of components
which might be $\Omega(n/\mu_i)$ in phase $i$.

After the above described round, each node $v$ \tj{knows $E_{F^v,\mu}$,} the set of all relevant edges
incident to its fragment. Thus, the set $E_{F,\mu_i}$ is computed
individually (and locally) by each node of $F$ (for each fragment $F$).
The choice of the sequence $\mu_i$ guarantees that each growable
fragment has at least $\mu_i$ elements in phase $i$. Therefore,
$E_{F,\mu_i}$ might be broadcasted to the whole network 
(step \ref{gen:s2} of Alg.~\ref{RCAST_MSF})
in one round such that
each node of $F$ broadcasts one element of $E_{F,\mu_i}$.
The range $r$ is equal to $1$ in this round, while the
edge capacity is $\Theta(\log n)$.
%

%\tj{Summarizing, the algorithm from \cite{Lotker:2003:MCO:777412.777428} determines
%MSF in $O(\log\log n)$ rounds, while its range is $\Omega(n)$ and edge capacity
%is $\Omega(\log n\log\log n)$.}

%\comment{
Below, we summarize properties satisfied by MSF algorithm from \cite{Lotker:2003:MCO:777412.777428}.
\begin{corollary}
\label{t:lotker:algorithm}
%\cite{Lotker:2003:MCO:777412.777428} 
There exists a deterministic congested clique MSF algorithm which
works in $O(\log\log n)$ rounds with range $r=O(n)$ and 
edge capacity $O(\log n\log\log n)$.
\end{corollary}
%}

\subsection{Minimum spanning forest algorithm in $\rcast(n,2)$}
In this section we will show an implementation of Algorithm~\ref{RCAST_MSF} 
%in 
%$\rcast(n,2)$ model
%working 
in $O(\log\log n)$ rounds, \tj{which is also efficient with respect to the
range and edge capacity.} 
%
%The hardest part to implement is 
%
As we discussed above, the only part of the Lotker et al.\ \cite{Lotker:2003:MCO:777412.777428} implementation of Algorithm~\ref{RCAST_MSF} with large range is the selection of the set of the lightest relevant edges %incident to 
\tj{for the current fragments.} 
%thus we will focus mainly on efficient implementation of this part of the algorithm.
Therefore, in order to reduce the range without increasing round complexity,
it is sufficient to design a new version of this part of Algorithm~\ref{RCAST_MSF}
for the sequence $\mu_1=1$ and $\mu_{i}=\mu_{i-1}(\mu_{i-1}+1)$ for $i>1$.
We give such a solution in this section.
    
%    \par Each phase of the presented Algorithm \ref{RCAST_MSF} consists of three parts. In the first part we select relevant edges for each fragment - as a result each node knows the set of relevant edges of its fragment. In the second part $\mu_i$ relevant edges incident to each fragment are announced to the whole network, which is trivial as long as there are no more edges to announce than nodes in the fragment. In the third part each node calculates new fragments in local memory. 
    
%    \subsubsection{Selecting the lightest edges for each component}
    %\par Selecting $\mu$ relevant edges is crucial part in \rcast(n,2) implementation of this algorithm. 
		
%		In this section we will show, how it can be done using a constant number of \rcast(n,2) rounds with $O(1)$ total capacity of communication edges. 

First, observe that the set of $\mu$ lightest relevant edges incident to a fragment $F$ (i.e., $E_{F,\mu}$) is included in the union of $\mu$ lightest relevant edges incident to each node from $F$,
i.e., $E_{F,\mu}\subseteq\bigcup_{v\in F}E_{v,\mu}$. Thus, in order to determine
$E_{F,\mu}$, it is sufficient to distribute/broadcast
information about $E_{v,\mu}$ for each $v\in F$ among nodes of $F$.
This task corresponds to the local broadcast problem (see Section~\ref{sub:s:tools}).
More precisely, given a partition $\mathbb{F}=\{F_1,\ldots,F_k\}$ in phase $i$, each
$v\in F_j$ is supposed to broadcast the message $M^v$ of size $b_i=O(\mu_i\log n)$
(i.e., description of $\mu$ lightest relevant edges incident to $v$) to all nodes of $F_j$.
Using Proposition~\ref{p:local:broadcast}, we can solve this task in $O(1)$ rounds
with range $r=2$ and edge capacity $1$, provided 
\begin{equation}\label{e:broadcast}
|F_i|\mu_i\log n\leq n.
\end{equation}
However, for large fragments and/or large $\mu_i$, this inequality is not
satisfied. Therefore, we need a more general observation saying that $\mu$ lightest relevant edges
incident to a set $A$ (not necessarily a fragment) might be chosen from 
the sets of $\mu$ lightest edges incident to subsets $A_j$ forming a partition of $A$.
\begin{fact}\labell{f:mu:edges}
\tj{Let $\mathbb{F}$ be a partition of a graph in fragments and %$A$ be a subset of a fragment $F\in \mathbb{F}$.
%Moreover, 
let $A_1,\ldots,A_k$ be a partition of 
the set of nodes of a fragment $F\in\mathbb{F}$.
%$A$ and let $\mu$ be a natural number.
Then, for each $\mu\in\NAT$, $E_{F,\mu}\subseteq\bigcup_{j\in[k]}E_{A_j,\mu}$.}
\end{fact}
Using Fact~\ref{f:mu:edges} we compute $E_{F,\mu}$ for a large fragment in the following way.
The set $F$ is split into small groups and $\mu$ lightest relevant edges are selected for each group and knowledge about them is distributed among nodes of the group. Then, the leader of each group is chosen and the task is reduced to choosing $\mu$ lightest relevant edges among the sets of
$\mu$ edges known to the leaders. This reduces our problem to its another instance with smaller
size of nodes. Another issue to deal with is to set the value of $\mu_i$ not too 
large for each $i$, in order to satisfy (\ref{e:broadcast}).
The choice of parameters in Algorithm~\ref{a:edgeSelection} guarantees that
the task of selecting $\min\{\mu_i,n^{1/3}\}$ lightest relevant edges incident to each fragment is possible
in $O(1)$ rounds with edge capacity $1$.

    \begin{algorithm}[h]
      \caption{SelectEdges($\mu,\mathbb{F}$)\Comment{the algorithm for node $v$}}
      \label{a:edgeSelection}
      \begin{algorithmic}[1]
			\State $\mu'\gets \min\{n^{1/3},\mu\}$
			\State $n_{\max}\gets n^{1/3}$
			\For{each $F\in\mathbb{F}$ and each $v\in F$ simultaneously}
				\State $A\gets$ the nodes of $F$
				\State $M^v\gets \mu'$ lightest relevant edges incident to $v$
				%\Repeat
				%\While{$|A|>1$}
				\If{$|A|\mu'\log n>n$}	
					%\State $k\gets \lfloor|A|/n_{\max}\rfloor$
					\State $n'\gets \frac{|A|}{\mu'\log n}$
					\State $k\gets \lceil|A|/n'\rceil$
					%\State split $A$ into disjoint subsets $A_1,\ldots,A_k$ such that
						%$|A_i|=n_{\max}$ for $i<k$ and $|A_k|\le n_{\max}$
					\State split $A$ into disjoint subsets $A_1,\ldots,A_k$ such that
						$|A_i|=n'$ for $i<k$ and $|A_k|\le n'$
                    \For{each $A_i$ simultaneously}
                      \State LocalBroadcast$(A_i, A_i, \mu'\log n)$
                    \EndFor
					\State let $A_j$ denote the set which contains $v$
					\State $M^v\gets\mu'$ lightest edges incident to $A_j$
					%\If{|A|
					\If{ID$(v)=\min\{\text{ID}(u)\,|\, u\in A_j\}$} %for some $j\in[k]$}
						\State $M^v\gets\mu'$ lightest edges incident to $A_j$
					\Else
						\State $v$ is removed from $A$
					\EndIf
				%\EndWhile
				%\Until{$|A|\le n_{\max}$}
				\EndIf
				\State LocalBroadcast$(A, F, \mu'\log n)$
				\State $v$ determines $E_{F,\mu'}$ on the basis of received messages\Comment{see Fact~\ref{f:mu:edges}}
      %\State {each node $v$ selects $\mu$ relevant edges} 
			%\Comment those edges can be encoded on $\mu \log n$ bits
      %\State \textbf{each node $v$ sends those edges to all other nodes in its group using partition\_broadcast}
      %\State \textbf{each node $v$ selects $\mu$ lightest relevant edges for its group}
      %\State \textbf{each node $v$ sends $i$th lightest relevant edge to all nodes in a fragment if its id is the $i$th smallest id among nodes in its group using partition\_broadcast}
      %\State \textbf{each node $v$ from all received edges selects $\mu$ relevant for the whole fragment}
			\EndFor
      \end{algorithmic}
    \end{algorithm}
    
\begin{proposition}\labell{p:select:edges}
Algorithm~\ref{a:edgeSelection} determines the set $E_{F,\mu'}$ 
lightest relevant edges incident to each fragment $F\in\mathbb{F}$
in $O(1)$ rounds with edge capacity $1$ and range $r=2$, where
$\mu'=\min\{n^{1/3},\mu\}$.
Moreover, $E_{F,\mu'}$ is known to each $v\in F$ for each $F\in\mathbb{F}$ at the end of an execution.
\end{proposition}		
\begin{proof}
Assume that $n$ is large enough to satisfy $n^{1/3}>\log n$.
First observe that the inequality $|A|\mu'\log n\le n$
is satisfied when the last step of the algorithm is executed.
If $|F|\mu'\log n\le n$ then the claimed inequality holds,
since $|A|=|F|$ in this case. Otherwise, the size of $A$
is reduced to 
$$k=\frac{|A|}{|A|/(\mu'\log n)}=\mu'\log n\le n^{1/3}\log n < \frac{n^{2/3}}{\log n}.$$
The choice of $n'$ guarentees also that $|A_j|\mu'\log n\le |F|\le n$
for each $j\in[k]$.
%Execution of LocalBroadcast are the only parts of the algorithm which requires communication between nodes. As $|A_i|\le n^{1/3}$ for each $i$ and $\mu'\le n^{1/3}\log n\le n^{2/3}$
%(for sufficiently large $n$), each 
Also, all fagments are pairwise disjoint, and all sets $A_j$ are pairwise disjoint (as a disjoint subsets of fragments). Thus, all execution of LocalBroadcast last $O(1)$ rounds
with edge capacity $1$, by Proposition~\ref{p:local:broadcast}.

By Fact~\ref{f:mu:edges}, the algorithm determines $\mu'$ lightest
relevant edges for elements of partitions of $F$ and eventually determines
$\mu'$ lightest relevant edges for each $F\in\mathbb{F}$, i.e.,
$E_{F,\mu'}$. For each $F\in\mathbb{F}$, the set $E_{F,\mu'}$ is 
known to all element of $F$ at the end of the execution of the algorithm, thanks to LocalBroadcast executed in the last step
of the algorithm. 
\end{proof}

\comment{
    \par Let us formulate for which $\mu_i$ we can implement Algorithm \ref{RCAST_MSF} in \rcast(n,2). Let us denote by $s_i$ the size of the smallest growable fragment. Then we will set $\mu_i = \min(s_i, (\frac{n}{\log n})^{1/3})$.
     As long as the size of the smallest growable fragment is smaller than $(\frac{n}{\log n})^{1/3}$ algorithm works as original version presented in \cite{Lotker:2003:MCO:777412.777428}. When the size of the smaller growable fragment will be larger than $(\frac{n}{\log n})^{1/3}$, then in each phase, by Lemma \ref{lotkerLemma}, we will increase it at least $(\frac{n}{\log n})^{1/3}$ times. Thus, after constant number of such phases we will find the minimal spanning forest.
}
Using Algorithm~\ref{a:edgeSelection} in the template described by Algorithm~\ref{RCAST_MSF},
we obtain the following result.
\begin{lemma}\labell{l:msf:det:range}
%%%
Assume that $\mu_1=1$ and $\mu_{i}=\min\{n^{1/3},\mu_{i-1}(\mu_{i-1}+1)\}$ for $i>1$.
Then, an implementation of Algorithm~\ref{RCAST_MSF} using the procedure
SelectEdges from Algorithm~\ref{a:edgeSelection} solves the MSF problem
in $O(\log\log n)$ rounds with range $r=2$.
\end{lemma}
\begin{proof}
After an execution of SelectEdges, a designated node $v\in F$ for each fragment $F$
knows $\mu_i$ edges which should be broadcasted to all nodes in step ~\ref{gen:s2}.
The definition of the sequence $\mu'_i$ and Lemma~\ref{lotkerLemma} guarantee
that the smallest size of a fragment at the beginning of phase $i$ is
at least $\mu'_i$.
Using these facts, one can implement step~\ref{gen:s2} of Algorithm~\ref{RCAST_MSF}
in two rounds. In round 1, that the node $v\in F$ which knows $E_{F,\mu'}$
sends the $j$th edge from $E_{F,\mu'}$ to the $j$th element of $F$. In round 2,
each node broadcasts an edge received in round 1 to the whole network.
Thus, each iteration of the while-loop works (i.e., each phase) works in $O(1)$ rounds with range $r=2$.

It remains to determine the number of iterations of the while-loop
(i.e., the number of phases).
For some $i=O(\log\log n)$ we get $\mu_i\ge n^{1/3}$.
The smallest size of a (growable) component
is larger than $n^{1/3}$ after $i=O(\log\log n)$ phases.
For $j>i$, the smallest size of a growable component is increased 
(at least) $n^{1/3}$ times in the $j$th round. This results in size
$n$ for the phase $i+2$ and shows that the algorithm works in $O(\log\log n)$
rounds.
\end{proof}

    \subsubsection{Reduction of total edge capacity}
%    \par Lotker et at.\ \cite{Lotker:2003:MCO:777412.777428} solution $\mathbb{L}$ to the MSF problem used messages of size $\Theta(\log n)$ in each phase. Thus, the best limitation on total capacity of communication edges is $B(\mathbb{L}) \in O(\log n \log \log n)$. 

Our solution for the MSF from Lemma~\ref{l:msf:det:range} reduces the range $r$ to $2$, but
each phase requires sending $\Theta(\log n)$ bits by some nodes, because weights of
some edges are transmitted by nodes in step~\ref{gen:s2} of Alg.~\ref{RCAST_MSF}.
%which gives
%total edge capacity of the order $\
In order to reduce (total) edge capacity, we modify the sequence $\{\mu_i\}$ again
to make it possible that step~\ref{gen:s2} of Alg.~\ref{RCAST_MSF} requires
$O(1)$ edge capacity for large fragments and edge capacities summarize to $O(\log n)$ for
small fragments. More precisely, let
$$\mu_i=\left\{
\begin{array}{lcl}
1 & \mbox{ for } & i\le 2\log\log n\text{ (Stage 1)}\\
\min\{\mu_{i-1}^2/\log n, n^{1/3}\} & \mbox{ for } & i> 2\log\log n\text{ (Stage 2)}
\end{array}
\right.
$$

%1$ for $i\in[2\log\log n]

Then, we implement Alg.~\ref{RCAST_MSF} as described in Lemma~\ref{l:msf:det:range} for
the new sequence $\{\mu_i\}_i$. One can verify that executions of SelectEdges can still
be implemented in $O(1)$ rounds with capacity $1$. 
%Moreover, the inequality $\mu_i\ge n^{1/3}$ is achieved for $i=O(\log\log n)$.
However, to reduce also total edge capacity of the whole
algorithm we change implementation of the part, where the edges from $E_{F,\mu}$
are announced for each $F$ to the whole network (step~\ref{gen:s2} of Alg.~\ref{RCAST_MSF}). 
Using Lemma~\ref{lotkerLemma}, one can observe that the size of the smallest growable fragment is
\begin{itemize}
\item
at least $2^{i-1}$ at the beginning of 
phase $i\le 2\log\log n$;
\item
at least $\mu_i$ at the beginning of 
phase $i>2\log\log n$.
\end{itemize}

    \par In a phase of $i\le 2\log \log n$ phases %we will set $\mu_i$ equal to $1$. 
		%The procedure \textbf{select\_edges} remains unchanged. In the procedure \textbf{announce\_edges} 
		each fragment $F$ has to broadcast a message $M^F$ of $\Theta(\log n)$ bits describing the lightest relevant edge incident to $F$. We split this message into $|F|$ fragments, each
of length $O(\frac{\log n}{|F|})$.

%    \par As we executed $2\log \log n$ phases in the stage one, the size of the smallest growable fragment in the second stage is at least $2^{2\log \log n} = \log ^2 n$. Let $s_i$ be the size of the smallest growable fragment in the $i$th phase. In the second phase we will use algorithm \ref{RCAST_MSF} with $\mu_i$ equal to $\min(\frac{s_i}{\log n}, (\frac{n}{\log n})^{\frac{1}{3}})$. Again, procedure \textbf{select\_edges} remains unchanged. 
%As for announcing relevant edges, 
For $i>2\log\log n$ and a fragment $F$
we want to broadcast a description of $\frac{|F|}{\log n}$ edges, which consists of $O(\frac{|F|}{\log n} \log n)=O(|F|)$ bits. In order to do that it is enough that each node announces $O(1)$ bits to the whole network, cf.\ Proposition~\ref{p:broadcast}. 
    
    By analyzing this algorithm, we will prove the following result. %theorem \ref{bandwidthTheorem}
    \begin{theorem}
      \label{t:bandwidthTheorem}
      It is possible to calculate the minimum spanning forest in $O(\log \log n)$ rounds and with total capacity of communication edges $O(\log n)$ and range $r=2$.
    \end{theorem}
    \begin{proof}
    \noindent\textbf{Number of rounds.}
    The first stage consists of $2 \log \log n$ rounds by definition. The second stage also consists of $O(\log \log n)$ rounds, however, we need a slightly more detailed analysis
to show this fact. 
    \par At the beginning of the second stage, the size of all growable fragments is at least  $\log^2 n$. Assume that the size of each growable fragment at the beginning of phase $i$ is at least
		$\mu_i$.
		Then, $\frac{\mu_i}{\log n}$ lightest relevant edges announced by each fragment satisfies $\frac{\mu_i}{\log n}\ge\mu_i^{1/2}$. Therefore, by Lemma \ref{lotkerLemma}, the size of the smallest growable fragment increases $\mu_i^{1/2}+1$ times
		in a phase. Thus, the size of the smallest growable fragment in the $i$th phase during the second stage is limited from below by $f_{i}$ defined as follows: $f_{1+2\log\log n} = \log^2 n$, $f_i = f_{i-1}^{3/2}$ for $i>1+2\log\log n$. For some $i \in \Theta(\log \log n)$, the size of the smallest fragment will be at least $n^{1/3}$.
Then, as shown in the previous section (Lemma~\ref{l:msf:det:range}), we obtain MSF after $O(1)$ additional phases.
%		
%		of size at least $(\frac{n}{\log n})^{\frac{1}{3}}$. After that, by Lemma \ref{lotkerLemma}, after $O(1)$ additional phases we will find MSF.
    
\noindent\textbf{Total capacity of communication edges.}
    In the first stage we have $O(\log \log n)$ phases, the size of the smallest growable fragment in the $i$th phase is at least $2^{i-1}$. 
Thus total capacity of communication edges of the first stage is $O(\sum\limits_i \frac{\log n}{2^i})=O(\log n)$. In the second stage we have $O(\log \log n)$ phases, each is implemented in $O(1)$ rounds with edge capacity $1$, thus total capacity of communication edges of those stages is $O(\log \log n)$. Therefore total capacity of communication edges of presented algorithm is $O(\log n  + \log \log n) = O(\log n)$.
    \end{proof}

	%\input{cc.tex}	
	%\section{Randomized connected components}
	
	\section{Randomized rcast algorithm for connected components}
	%\subsection{Randomized algorithm for connected components}

The fastest known randomized algorithm calculating Connected Components in the unicast model
%$\rcast(n,n)$ 
works in $O(\log^* n)$ communication rounds \cite{GhaffariParter2016}. The algorithm works in phases. 
At the beginning of each phase, a partition of an input graph into components is known to all nodes.
In a phase of the algorithm, the number of growable components drops from $\frac{n}{\log^2 x}$ to $\frac{n}{x}$, by simulating $\Theta(\log x)$ steps of the standard Boruvka's algorithm.  
%
%In order to make this simulation (of the Boruvka's algorithm) in $O(\log^* n)$ rounds, 
Each phase is implemented in $O(1)$ rounds.
The key tool to make it possible is a special kind of \emph{linear sketches}.
% is designed and implemented
%in the unicast congested clique.

In the following, we first describe the linear sketches of Ghaffari and Parter \cite{GhaffariParter2016}. 
Then, we shortly describe the $O(\log^* n)$ algorithm for connected components in the unicast
congested clique \cite{GhaffariParter2016}. In the next part, we present an implementation of the algorithm in the $\rcast(n,2)$ model. Finally, we provide version of the algorithm with optimal total edge capacity and range $2$.

%%%%%%%%%
\subsection{Linear sketches}
In order to build Parter-Ghaffari's sketches for a graph with $n$ nodes, a preprocessing
is necessary. During the preprocessing, each (prospective) edge $(u,v)$ is assigned an ID of size $O(\log n)$,
based on a random seed of size $O(\log n)$.
In order to build sketches for a given graph $G(V,E)$ with $n$ nodes and a parameter $x\le n$, the sets
$E_1, E_2,\ldots, E_{10\log x}$ included in $E$ are \tj{chosen such that each
edge $e\in E$ belongs $E_j$ with probability $1/2^j$ and all random choices are independent}.
For $v\in V$ and $A\subset V$, let $E_j(v)$ be the set of elements of $E_j$ incident to $v$
and let $E_j(A)$ be the set of elements of $E_j$ incident to $A$, i.e.,
$E_j(A)=\{\{u,v\}\in E_j\,|\, u\in A, v\not\in A\}$.
\newcommand{\sketch}{\text{sketch}}
Then, $\sketch(\mathbb{X})$ ($\mathbb{X}$ may be a set or a single node) is a table consisting of $10\log x$ rows, each row contains a bit string of length $O(\log n)$. 
The $j$th row of $\sketch(\mathbb{X})$ is the xor of IDs of all elements of $E_j(\mathbb{X})$.
The sequence of $\log x$ sketches for a set or a node will be called its \emph{multi-sketch}. 
\tj{Thus, a multi-sketch might be seen as a table consisting of $10\log^2x$ rows.
By $\sketch_r(A)$ and $\text{multi-sketch}_r(A)$ we denote the $r$th row of a sketch and
a mutli-sketch of $A$, respectively.}

\tj{Below, we give the key properties of sketches for design of distributed algorithms
for graph connectivity.}
\comment{As shown in \cite{GhaffariParter2016}, a sketch of a component with $O(x^5)$ outgoing
The key property of sketches useful in design of distributed algorithms is that the sketch
of a set might be easily obtained from the sketches of its subsets.}
\begin{proposition}\labell{f:xor}\cite{GhaffariParter2016}
1.~It is possible to determine an edge $\{u,v\}$ such that $u\in A$ and $v\not\in A$ from a sketch of $A\subset V$ with
probability $\Omega(1)$,
provided the number of edges $\{u,v\}$ such that $u\in A$, $v\not\in A$ is at most $x^5$.\\
2.~The sketch of a set $A=A_1\cup A_2\subset V$ for disjoint sets $A_1, A_2$ is equal to\\ $\sketch(A_1) \text{ xor } \sketch(A_2)$. %\text{XOR}_{v\in A}\sketch(v)$.
That is, the $i$th row of $\sketch(A)$ is equal to the xor of the $i$th row of 
$\sketch(A_1)$ and the $i$th row of $\sketch(A_2)$. 
\end{proposition}
Note that, by Proposition~\ref{f:xor}.2, the sketch of a set $A$ is equal to
the xor of sketches of all elements of $A$.

\subsection{Ghaffari-Parter $O(\log^* n)$ connected components algorithm}

Ghaffari-Parter algorithm for connected components in the unicast congested clique works
in $O(\log^* n)$ phases, each phase consists of $O(1)$ rounds.
At the beginning of a phase, a partition $\mathbb{C}$ of an input graph into $O(n/\log^2 x)$ (growable)
components is known to each node. As a result of the phase, the number of components
is reduced to $O(n/x)$, whp.
During the phase (see Algorithm~\ref{alg:logstar} for the pseudocode):
\begin{enumerate}[label=(\roman*), noitemsep,topsep=0pt]
\item
multi-sketches are computed for each component and sent to the leader node $u^*$;
\item
the leader $u^*$ locally simulates $\log x$ steps of the Boruvka's algorithm, using
obtained \tj{multi-sketches} of components;
\item
the leader distributes (with help of other nodes) information about new partition into components;
\item \label{a:rand:rand}
each node $v$ broadcast a random edge $\{u,v\}$ such that $C^u\neq C^v$ and a partition is updated using the broadcasted
edges.\footnote{Random edges are necessary in order to deal with components with degree $>x^5$, because sketches do not help much to find their neighbors.}
\item \tj{non-growable components are deactivated.}
\end{enumerate}
    \begin{algorithm}[H]
      \caption{CCLogstar\Comment{the algorithm for a node $v\in C_k$}}
      \label{alg:logstar}
      \begin{algorithmic}[1]
	\State $\mathbb{C}=\{\{v_1\},\ldots,\{v_n\}\}$
	\While{$\mathbb{C}\neq\emptyset$}\Comment{i.e., while there are active components}
		\State $x\gets\min\{y\,|\, |\mathbb{C}|<\frac{n}{10\log^2x}\}$\Comment{$\mathbb{C}$ is the number of growable components}
		\State \label{star:sketches:comp} Compute multi-sketches of all components from $\mathbb{C}$
		\State \label{star:sketches:send} Distribute the multi-sketches in the network
		\State \label{star:boruvka} Update %the partition 
		$\mathbb{C}$ by simulating $\Theta(\log^*n)$ rounds of Boruvka's algorithm, using sketches.
		\State \label{star:real} Determine edges in the input graph which connect old components in the new ones
		\State \label{star:random} 
		Broadcast a random edge incident to each component, update $\mathbb{C}$ based on these edges.
		%Each $v$ sends a random edge from $\{\{u,v\}\in E\,|\, C^v\neq C^u\}$
		\State \label{star:deactivate} Deactivate (remove from $\mathbb{C}$) non-growable components.
	\EndWhile
      \end{algorithmic}
    \end{algorithm}

\noindent Step~\ref{star:real} of Alg.~\ref{alg:logstar} does not require any work in \cite{GhaffariParter2016}, since sketches are computed for actual edges of the input graph (this step  will become important in our new algorithm). Below, we shortly describe \tj{some other aspects of} an implementation of the above steps in the unicast congested clique in \cite{GhaffariParter2016}:
\begin{enumerate}[label=(\alph*),noitemsep,topsep=0pt]
\item
For each edge $\{u,v\}\in E$, the node with larger ID (say $u$) makes random choices determining
to which of the sets $E_1, E_2,\ldots, E_{10\log x}$ the edge $\{u,v\}$ belongs. 
\item
Each node $v$ computes individually its $\log x$ sketches. In order to make it possible,
$u$ sends (an encoding of) $\log^2 x$ bits to $v$ determining in which rows of multi-sketches
the edge $(u,v)$ is added,
\newcommand{\ID}{\text{ID}}
for each edge $\{u,v\}\in E$ such that $\ID(u)>\ID(v)$.
\item
Each component $C_i$ for $i\in[1,n/\log^2 x]$ has assigned a \emph{representative} set of 
$\log^2 x$ nodes
$V_i=\{v_{1+(i-1)y}, v_{2+(i-1)y}, \ldots, v_{iy}\}$, where $y=\log^2x$. Each node $v\in C_i$ sends
the $j$th row of multi-sketch of $v$ to the $j$th node of $V_i$. By xoring all obtained messages,
the $j$th node of $V_i$ knows the $j$th row of the multi-sketch of $C_i$ and sends it to $u^*$.
\item
After computing a new partition of nodes into components, $u^*$ sends to each $v\in V$ its new
component ID (\tj{the ``name'' of a component} might be, e.g., the smallest ID of a node inside the component). Then, each node
broadcasts its new component to all nodes. In this way each node knows a new partition into
components. 
\end{enumerate}
The above distributed implementation requires the range $r=2^{\log^2x}$ in part (b), since
$\log^2x$ bits are transmitted over each edge. As the unicast
model allows for %$r\leq n$ 
messages of length $\log n$ only, there is a problem if
$\log^2x>\log n$ (which appears in the last phase). 
In order to overcome this problem, the authors of \cite{GhaffariParter2016} argue that random distribution of transmitted
string guarantees that they can be encoded in $\log n$, whp.
The range $r$ of part (c) in the above implementation is $r=\log^2x$ 
and part (d) requires \tj{the range equal to the number of new components which might be $n/x$}.

The following result from \cite{GhaffariParter2016} implies
that Algorithm~\ref{alg:logstar} determines connected components
in $O(\log^*n)$ iterations of the while-loop, with high probability. 
\begin{lemma}\labell{l:phase:logstar}
An iteration of the while-loop Algorithm~\ref{alg:logstar} reduces the number of non-growable components from $n/\log^2x$ to
at most $n/x$, whp.
\end{lemma}

\subsection{Range efficient algorithm for connected components}
In order to implement a phase of Algorithm~\ref{alg:logstar} in the rcast model with
the range $r=2$ and in $O(1)$ rounds, we need a new method
of computing and distributing sketches.

%we use the following strategy to implement
%a phase, given 
Assume that a partition $\mathcal{C}$ into components is known
to all nodes at the beginning of a phase.
%
%\begin{itemize}[noitemsep,topsep=0pt]
%
%\item
%First, 
\tj{Consider a \emph{meta-graph}, whose nodes correspond to the current 
components, where $C_i, C_j$ are connected by
a \tj{\emph{meta-edge} iff there is an edge $\{u,v\}$ such }
that $u\in C_i$ and $v\in C_j$.  
From the ``point of view'' of nodes it means that $u$
and $v$ are connected by an edge iff $C^u$ and $C^v$ are
neighbors in the current meta-graph.

In our algorithm, the sketches are computed 
for the meta-graph and delivered to \textit{all} nodes. 
On the basis of the sketches, each node can simulate $\log x$ steps of the Boruvka's algorithm
\tj{on the meta-graph},
merging components into larger ones.}
%\item
After determining new larger components, information about the real edges
connecting merged input components \tj{(into new larger ones)}
are determined and broadcasted to all nodes.
\comment{To this aim, for each new component $C'$, a rooted tree $T_{C'}$
of components from $\mathcal{C}$ forming $C'$ is build.
The nodes $C_i, C_j\in\mathcal{C}$ can be connected by an 
edge
in $T_{C'}$ if a meta-edge connecting $C_i$ and $C_j$ has been
decoded during a phase. % (from sketches, or as random edges in
%step~\ref{star:random} of the algorithm).
%\item
Broadcasting random edges %from large degree components 
\tj{(step \ref{star:random} of the algorithm)}
requires one round
with range $r=1$, thus this part might remain unchanged.
%\end{itemize}
%
}
\noindent Below, we describe this strategy in more detail.

\paragraph{Computing sketches in a meta-graph}
%As in \cite{GhaffariParter2016}, 
For computing (and broadcasting) multi-sketches in a meta-graph, %works as follows.
%As in \cite{GhaffariParter2016}, 
each component $C_i$ is associated with 
a \emph{representative set} $V_i$ of size $\log^2x$.
In the first round, each node $v$ sends the bit $1$ to each element of $V_i$ for $i\in[\log^2x]$ iff 
$(v,u)\in E$ for some $u\in C_i$. Otherwise, $v$ sends $0$ to each node of $V_i$.
After such a round each node of $V_i$ knows all neighbors of $C_i$ in the meta-graph.
In order to compute and distribute a multi-sketch of $C_i$ in $O(1)$ rounds, we 
make the $j$th element of $V_i$ (say, $v_{i,j}$) responsible for the $j$th row of the multi-sketch of $C_i$.
For each edge $(C_i, C_{i'})$ such that $i>i'$, $v_{i,j}$ chooses with appropriate probability
\tj{(i.e., $1/2^{1+(j-1)\mod 10\log x}$)} whether this edge is included in the $j$th row of the multi-sketch.
In the second communication round $v_{i,j}$ sends 1 to $v_{i',j}$ when the edge is included
and $0$ otherwise. Using own random choices and messages received in both rounds, $v_{i,j}$
computes the $j$th row of the multi-sketch of $C_i$ and broadcasts it to the whole
network.
More precise description of the above strategy is presented in Algorithm~\ref{alg:sketches}. Below, we summarize efficiency of 
this algorithm.
\tj{\begin{proposition}\labell{p:sketches:our}
Assume that Algorithm~\ref{alg:sketches} is executed for
a partition of an input graph in at most $n/\log^2x$ components.
Then,
the algorithm
%Algorith~\ref{alg:sketches} 
determines multi-sketches of all nodes in the meta-graph
and broadcasts
them to the whole network in $O(1)$ rounds, with range $r=2$ and edge capacity $O(\log n)$.
\end{proposition}
}

    \begin{algorithm}[H]
      \caption{LinearSketches\Comment{the algorithm for a node $v\in C_k$}}
      \label{alg:sketches}
      \begin{algorithmic}[1]
			\State \label{sketches:begin} $y\gets 10\log^2x$
			\State Let $C_k$ be the component containing $v$
			\State $V_i\gets \{v_{i,1},\ldots,v_{i,y}\}$ for $i\in[n/y]$, where $v_{i,j}=v_{(i-1)y+j}$
			\State Let $v=v_{p,r}$
			\For{each $j\in[n/y]$}
				\If{$\{\{v,u\}\,|\, u\in C_j\}\neq \emptyset$} 
					\State $b_j\gets 1$
				\Else
					\State $b_j\gets 0$
				\EndIf
			\EndFor
			\State \label{sketches:round1} \textbf{Round 1:} $v$ sends $b_j$ to each node of $V_j$ for each $j\in[n/y]$
			\State $E(C_p)\gets \{(C_p,C_l)\,|\, $1$\text{ received from some } u\in C_l\}$
			\State \textbf{for} each $e=(C_p,C_l)\in E(C_p)$ such that $l>p$: set $b_l\gets 1$ with probability $1/2^{1+(r-1)\text{ mod } y}$, $b_l\gets 0$ otherwise
			\State \textbf{Round 2:} $v=v_{p,r}$ sends $b_l$ to $v_{l,r}$ for each $l\in [n/y]$
			\State \textbf{for} each $l<p$: set $b_l$ to the bit received in Round~2 from $v_{l,r}$
			\State multi-sketch$_r(C_p)\gets \text{xor}_{k\in[n/y]}b_k\cdot \text{ID}((C_p,C_k))$
			\State \textbf{Round 3:} $v_{p,r}$ sends multi-sketch$_r(C_p)$ to all nodes
			%\State 
      \end{algorithmic}
    \end{algorithm}

\paragraph{Determining real edges connecting merged components}
%
%PONIZEJ NIE ZMIENIALEM NA RAZIE, NIE JEST SKOORDYNOWANE Z TYM CO WYZEJ}
An offline simulation of the Boruvka's algorithm based on meta-edges derived from sketches
gives a new partition into components $\mathcal{C}'$. Each component $C'$ of this new
partition is a connected subgraph of the meta-graph, with meta-edges connecting elements
of $C'$ known to all nodes (determined by sketches). 
In order to determine real edges connecting \tj{elements of} $C'$, a rooted spanning tree for $C'$ is chosen
arbitrarily %(but deterministically) 
but in the same way by each node $v\in C'$.
For each node $v$, if $v$ is adjacent to an (``real'') edge $(u,v)$ such that 
$C^u$ is the parent of $C^v$ then $v$ chooses such edge arbitrarily.
Then, in one round, $v$ broadcasts such chosen edge to the whole
network.
%
%Summarizing, we have the following result.

Let CCLogstarR be a variant of Alg.~\ref{alg:logstar}
where the steps~\ref{star:sketches:comp} and \ref{star:sketches:send}
of the algorithm is implemented through Algorithm~\ref{alg:sketches}. 
In order to decode the real edges in the input graph 
between joined components, we use the method described above
which requires one round with capacity $O(\log n)$ and
range $1$ in order to make step~\ref{star:real} of 
our implementation of Alg.~\ref{alg:logstar}.
\begin{lemma}\labell{l:logstarr}
Algorithm CCLogstarR identifies the connected components of the input
graph in $O(\log^* n)$ communication rounds in the
rcast$(n,2)$ model, with high probability.
\end{lemma}
\tj{ %CZY TEN LEMAT WYMAGA KOMENTARZA? POWOLANIA SIE NA GH-PARTER?}
\begin{proof}
Prop.~\ref{p:sketches:our} implies that, in each phase, all nodes will know
sketches of all components and each node can perform locally
$\Theta(\log x)$ 
steps of the Boruvka's algorithm on the meta-graph in the way
described in \cite{GhaffariParter2016}.  
Moreover, the real edges showing connectivity of components
are decoded as described above in one round in each iteration
of the while-loop.
Thus, the algorithm determines connected components in
$O(\log^* n)$ rounds with range $2$, by Lemma~\ref{l:phase:logstar}.\footnote{One doubt which might appear here is that the authors
of \cite{GhaffariParter2016} deal with real edges in all phases,
while our implementation just considers meta-edges between
components. However, the reduction of the number of components
in a phase holds whp for arbitrary graph, thus also for a meta-graph
of components.}
\end{proof}
}

\subsection{Reduction of total edge capacity}
%\tj{The} presented algorithm does not have an optimal total capacity of communication edges. 
In this section we will show that it is possible to achieve the optimal edge capacity 
\tj{$O(\log n)$ without increasing the range or round complexity
of CCLogstarR.}
%increasing the number of rounds by a constant factor. 
More precisely, we show the following theorem.
%Theorem~\ref{t:logstar}.
\begin{theorem}\labell{t:logstar}
There is a randomized algorithm in the rcast$(n,2)$ congested clique that identifies the connected components of the input
graph with total edge capacity $O(\log n)$ in $O(\log^* n)$ communication rounds, with high probability.
\end{theorem}
\begin{proof}
\par There are three steps of our \rcast(n,2) implementation of CCLogStar (see Lemma~\ref{l:logstarr}), which require $\Omega(\log n)$ bits in each phase. The first is announcing multi-sketches
(step~\ref{star:sketches:send} of Alg.~\ref{alg:logstar}, implemented as Round~3 in Alg.~\ref{alg:sketches}), the second is determining and announcing real edges
connecting components (step~\ref{star:real} of Alg.~\ref{alg:logstar}, required because of the fact that sketches are computed with
respect to the meta-graph -- see the description of CCLogstarR) and the third is announcing a random edge (step~\ref{star:random} of Alg.~\ref{alg:logstar}). 
%Lemma~\ref{l:logstarr}).

\par In order to announce multi-sketches with smaller edge capacity, we slightly change the whole algorithm. In each phase we will select %such 
$x$ %such that the average size of a component is $\Theta(\frac{n}{\log^3 x})$ instead of
equal to $\min\{y\,|\, |\mathbb{C}|<\frac{n}{10\log^3x}\}$ instead of
$\min\{y\,|\, |\mathbb{C}|<\frac{n}{10\log^2x}\}$.
% $\Theta(\frac{n}{\log^2 x})$. 
Therefore representative sets in Alg.~\ref{alg:sketches} can now have size $10\log^3x$. In order to compute sketches in Alg.~\ref{alg:sketches}, we will use only the $10\log^2 x$ nodes from each representative set $V_i$ as presented before. The only part requiring a change is announcing the meta-sketch to the whole network (\textbf{Round~3} of Alg.~\ref{alg:sketches}). 
Consider a representative set $V_i$ of size $10\log^3x$ and its subset $V'_i\subset V_i$ of size $10\log^2x$ such that the $j$th node of $V'_i$ has computed the $j$th row of the multi-sketch
of $C_i$ of $O(\log n)$ bits. Then, using the local broadcast primitive (Proposition~\ref{p:local:broadcast}) for $T=V'_i$, $R=V_i$ and $b=O(\log n)$, the whole multi-sketch$(C_i)$ can be
distributed to all nodes of $V_i$ in $O(1)$ rounds with edge capacity $1$, since
%
%each $v'\in V'_i$ might broadcast, using Proposition~\ref{p:local:broadcast} as 
$\log n \log^2 x \in O(n)$. Next, multi-sketch$(C_i)$ can be announced to the whole network 
by the global broadcast procedure with $S=V_i$ and $b=O(\log^2x\log n)$.
By Proposition~\ref{p:broadcast}, this task can be done in one round with range $1$
and edge capacity $O(\frac{\log^2x\log n}{\log^3 x})=O(\frac{\log n}{\log x})$.

\comment{For each representative set, $\Theta(\log^2 x)$ nodes with smallest ids have one row of multi-sketch consisting of $O(\log n)$ bits. Therefore, each of those nodes can announce the row of the multi-sketch to all nodes in the representative sets, using Proposition~\ref{p:local:broadcast} as $\log n \log^2 x \in O(n)$. Then multi-sketches can be announced to the whole network using edge capacity $O(\frac{\log n}{\log x})$.
}

\par In order to determine and announce real edges connecting ``old'' components into ``new'' ones (step~\ref{star:real} of Alg.~\ref{alg:logstar}), we execute the following procedure. Let $\mathbb{C}$ denote the ``old'' partition into components before step~\ref{star:boruvka} on the meta-graph and let $\mathbb{C}'$ denote the ``new'' partition after that step.
After determining $\mathbb{C}'$ and decoding meta-edges from sketches locally (the simulation of Boruvka's algorithm in step~\ref{star:boruvka} of Alg.~\ref{alg:logstar}), all nodes build locally a forest $\mathbb{F}$ of rooted trees (using disclosed meta-edges), connecting old components from $\mathbb{C}$ in the new ones from
$\mathbb{C}'$. Then, in a separate round,
each node $v$ sends the bit $B=1$ iff $v$ is incident to an edge connecting its old component
$C^v$ to the parent of $C^v$ in the appropriate tree; $v$ sends $B=0$ otherwise. 
Consider $C$ which is not the root of a tree in $\mathbb{F}$.
Based on transmitted bits, the node $v_C$ is chosen as the one with the smallest ID among elements of $C$ which sent $B=1$ (i.e., among nodes incident to edges with an endpoint in the parent of $C$ in $\mathbb{F}$). 
%
%will choose a node with the smallest ID which is incident to \tj{a real edge of interest}. It can be done in one round in $O(1)$ bits per communication link, as each node with real edge to announce can send one bit of information determining, whether it has such edge, to all other nodes in its group. All nodes know which nodes inside component have real edges to announce, thus its simply to select one with the smallest ID. 
Then, $v_C$ announces the real edge connecting $C$ and the parent of $C$ to all nodes in the representative set of $C$ by local broadcast (Proposition~\ref{p:local:broadcast}) with $T=\{v_C\}$, $R=V^C$ and $b=\log n$,
where $V^C$ is the representative set of the component $C$.
Then, the global broadcast procedure is applied with $S=V^C$ and $b=\log n$
(see Proposition~\ref{p:broadcast}).
In this way a real edge connecting $C$ with the parent of $C$ is announced
to the whole network . Thus we implement step~\ref{star:real}
of Alg.~\ref{alg:logstar} in $O(1)$ rounds with edge capacity $2$.

\par In order to select and announce a random meta-edge incident to each component  (step~\ref{star:random} of Alg.~\ref{alg:logstar}), we use the strategy from the previous
paragraph. Namely, a random edge incident to the component $C$ is chosen by the node with
the smallest ID in $C$. Then, this edge is broadcasted to the whole network using 
the local broadcast and the global broadcast primitives, with help of the representative
set of $C$. The only issue here is that a node $v$ knows only edges incident to $v$, not
the edges incident to the whole component $C^v$. However, each node can learn
neighborhood of its component using Round~1 from Alg.~\ref{alg:sketches}.
Thus, the choice of a random edge is preceded by such a round (described by steps~\ref{sketches:begin}--\ref{sketches:round1} of Alg.~\ref{alg:sketches}).
Thus, we chose a random edge incident to each component in the new meta-graph; we can decode the real
edges corresponding to the chosen meta-edges as described in the previous paragraph.
In this way, we implement step~\ref{star:random}
of Alg.~\ref{alg:logstar} in $O(1)$ rounds with edge capacity $1$ and range $2$.

%for each component one node in its representative set selects a random edge in a component graph (as before). Then instead of announcing it to the whole network, the edge is announced inside this representative set using Proposition~\ref{p:local:broadcast}, and then each representative set announces selected meta-edge to the whole network using Proposition~\ref{p:broadcast}.

\par Summarizing, we obtain an algorithm which determines connected components 
in $O(T(n))$ rounds whp, where $T(n)$ is equal to the smallest $i$ such that
$f_i\ge n$ for the sequence $f_1=c$ and $f_i=2^{f_{i-1}^{1/3}}$ for a constant $c\ge 1$.\footnote{One can make $O(\log c)$ steps of Boruvka's algorithm at the beginning, in order to start from the components of size $\ge c$.} 
One can easily verify
that $f_i\geq n$ for $i=O(\log^* n)$.
As we showed above, the range of our algorithm is $r=2$.
The edge capacity of phase $i$ is $O\left(\frac{\log n}{\log x}\right)$, where
$x\geq f_i$ whp. As $f_i>2^i$ for each $i$ if the constant $c=f_1$ is large enough, the total edge capacity is 
$O(\sum\limits_i \frac{\log n}{2^i}) = O(\log n)$.
\end{proof}

%Given a partition of a graph into components and a parameter $x$,
%a sketch of a component set is a 

	\section{Conclusions}

%\comment{ %% krotka wersja jesli walczymy o miejsce
We have shown the first sub-logarithmic algorithm for connected components in the broadcast
congested clique. Moreover, we provided efficient rcast$(n,2)$ implementations of the
deterministic MSF algorithm \cite{Lotker:2003:MCO:777412.777428} and randomized
algorithm for connected components \cite{GhaffariParter2016}. Both implementations are
not only time efficient but also optimal with respect to maximal edge capacity of communication
edges. An interesting research problem arising from these results is to determine a relationship between adaptiveness (the
number of rounds) and total capacity of communication edges.
Moreover, it is still not known whether MSF can be computed in $o(\log n)$ rounds or connected components can be computed in $o(\log n/\log\log n)$ rounds in the broadcast congested clique.

%}

\comment{ %poprzednia wersja
In this paper, we have shown the first sub-logarithmic algorithm for connected components in the broadcast
congested clique. Moreover, we provided implementation of $O(\log \log n)$ round deterministic algorithm \cite{Lotker:2003:MCO:777412.777428} finding MSF in \rcast(n,2), which led to version of algorithm with total capacity of communication edges $O(\log n)$, which is optimal. Additionally we provided \rcast(n,2) implementation of randomized connected components algorithm \cite{GhaffariParter2016}, however we still do not know whether it is possible to find MST in $O(\log^* n)$ rounds in \rcast(n,2) and whether it is possible to find connected components in $O(\log^* n)$ rounds and optimal total capacity of communication edges.
\par An interesting direction arising from these results is to determine a relationship between adaptiveness (the
number of rounds) and total capacity of communication edges.
}

\comment{Another interesting research direction concerns the impact of the message size on the round complexity of problems in the congested clique. Recently, it was shown that connectivity can be solved in the broadcast congested clique by a randomized algorithm in just one round for message size $O(\log^3n)$ \cite{AhnGM12} and by a deterministic algorithm in three rounds for message size $O(\sqrt{n}\log n)$ \cite{DBLP:journals/corr/MontealegreT16}.}

%	\bibliographystyle{abbrv}
%	\bibliography{references}

	\section{Appendix}
	\iffull
\else
\subsection{Proof of Fact~\ref{simulationTheorem}}
The local computation part of the algorithm $A$ stays the same. As for communication part, we can split each message of original protocol into blocks of size $\lfloor \log_r \rfloor$ and sent them in separate rounds.  Therefore, if in one round protocol sent message of size $\log n$, after $\frac{\log n}{\lfloor \log_r \rfloor} = O(\log_r n)$ rounds whole message would be sent to receiver, with no more than $2^{\lfloor \log_r \rfloor} \leq r$ messages per round.
\fi

\end{document}